\documentclass[a4paper,twoside,11pt]{article}
\usepackage[margin=2.9cm]{geometry}
\usepackage[affil-it]{authblk}
\usepackage{amsmath,amsfonts,amsthm,amssymb}
\usepackage[font={small}]{caption}
\usepackage[font={small}]{subfig}
\usepackage{graphicx}
\usepackage{paralist}
\usepackage{times}
\usepackage{xspace}

\graphicspath{{./img/}}

\title{Planar Octilinear Drawings with One Bend Per Edge}

\author[1]{Michael A. Bekos%
\thanks{Electronic address: \texttt{bekos@informatik.uni-tuebingen.de}}}

\author[2]{Martin~Gronemann%
\thanks{Electronic address: \texttt{gronemann@informatik.uni-koeln.de}}}

\author[1]{Michael~Kaufmann%
\thanks{Electronic address: \texttt{mk@informatik.uni-tuebingen.de}}}

\author[1]{Robert~Krug%
\thanks{Electronic address: \texttt{krug@informatik.uni-tuebingen.de}}}

\affil[1]{Wilhelm-Schickhard-Institut f\"ur Informatik, Universit\"at T\"ubingen, Germany}
\affil[2]{Institut f\"ur Informatik, Universit\"at zu K\"oln, Germany}

\date{}

\newtheorem{lemma}{Lemma}
\newtheorem{theorem}{Theorem}
\newtheorem{definition}{Definition}
\newcommand{\skel}[1]{G^\textit{skel}_{#1}}
\newcommand{\Vpert}[1]{V^\textit{pert}_{#1}}
\newcommand{\Epert}[1]{E^\textit{pert}_{#1}}
\newcommand{\Vskel}[1]{V^\textit{skel}_{#1}}
\newcommand{\Eskel}[1]{E^\textit{skel}_{#1}}
\newcommand{\pert}[1]{G^\textit{pert}_{#1}}
\newcommand{\gdeg}[1]{\textit{deg}(#1)}
\newcommand{\pdeg}[2]{\textit{deg}^{\textit{pert}}_{#1}(#2)}
\newcommand{\sdeg}[2]{\textit{deg}^{\textit{skel}}_{#1}(#2)}
\newcommand{\redge}[1]{\textit{ref}(#1)}
\newcommand{\poles}[1]{\mathcal{P}_{#1}}
\newcommand{\IPGeo}{IP-\ref{ip:1}\xspace}
\newcommand{\IPFix}{IP-\ref{ip:2}\xspace}
\newcommand{\IPPort}{IP-\ref{ip:3}\xspace}
\begin{document}
\maketitle

%=================================================================
\begin{abstract}
In \emph{octilinear drawings} of planar graphs, every edge is drawn
as an alternating sequence of horizontal, vertical and diagonal
($45^\circ$) line-segments. In this paper, we study octilinear
drawings of low edge complexity, i.e., with few bends per edge. A
$k$-planar graph is a planar graph in which each vertex has degree
less or equal to $k$. In particular, we prove that every 4-planar
graph admits a planar octilinear drawing with at most one bend per
edge on an integer grid of size $O(n^2) \times O(n)$. For 5-planar
graphs, we prove that one bend per edge still suffices in order to
construct planar octilinear drawings, but in super-polynomial area.
However, for 6-planar graphs we give a class of graphs whose planar
octilinear drawings require at least two bends per edge.
\end{abstract}
%=================================================================

%=================================================================
\section{Motivation}
\label{sec:introduction}
%=================================================================

Drawing edges as octilinear paths plays a central role in the design
of metro-maps (see e.g., \cite{HMN06,NW11,SROW11}), which dates back
to the 1930's when Henry Beck, an engineering draftsman, designed
the first schematic map of London Underground using mostly
horizontal, vertical and diagonal segments; see
Fig.\ref{fig:beck1933}. Laying out networks in such a way is called
\emph{octilinear graph drawing}. More precisely, an \emph{octilinear
drawing} of a (planar) graph $G=(V,E)$ of maximum degree eight is a
(planar) drawing $\Gamma(G)$ of $G$ in which each vertex occupies a
point on the integer grid and each edge is drawn as a sequence of
alternating horizontal, vertical and diagonal ($45^\circ$)
line-segments. For an example, see Fig.\ref{fig:4p_example_large} in
Section~\ref{sec:sample}.

\begin{figure}[t]
    \centering
    \includegraphics[width=.7\textwidth]{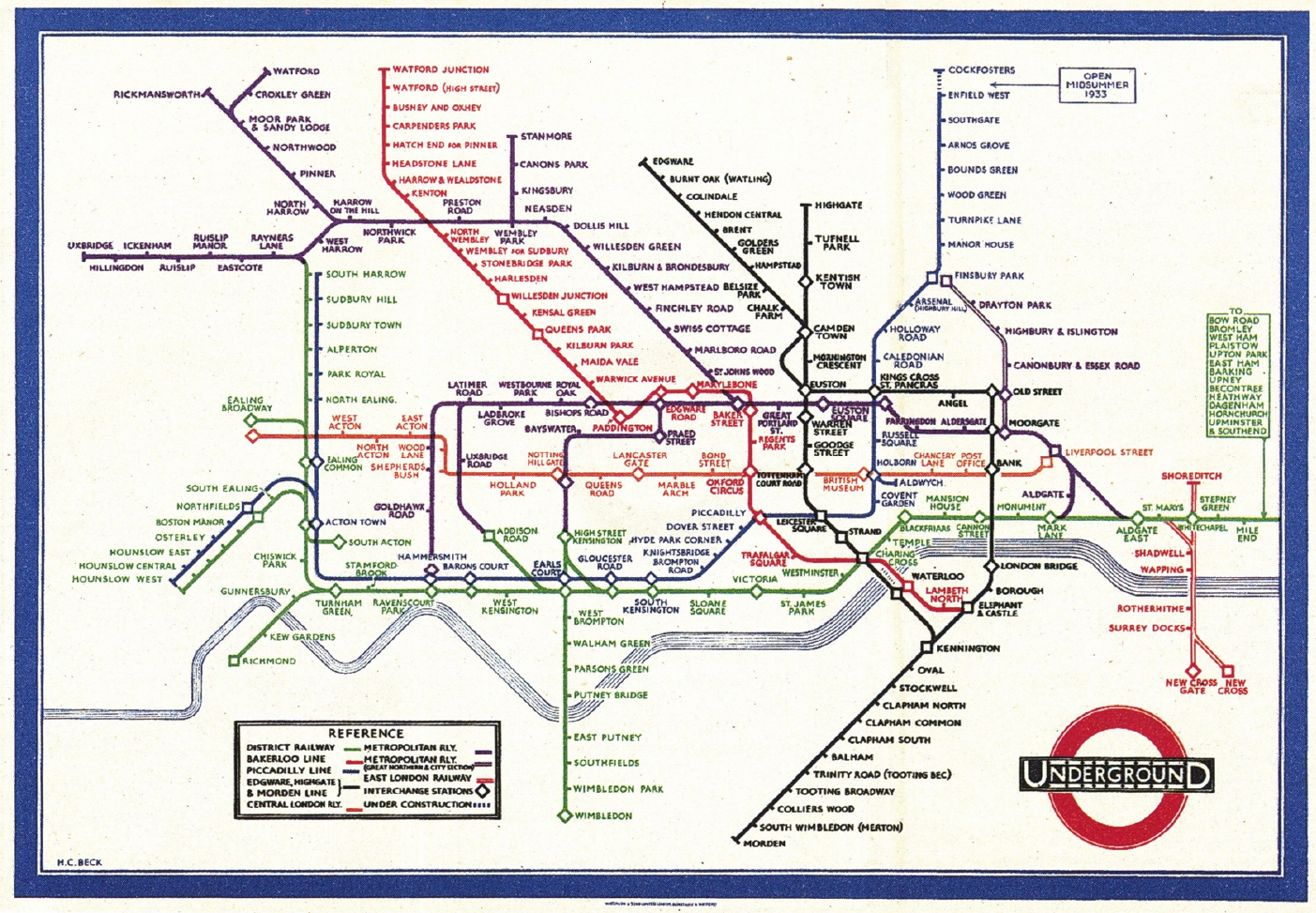}
    \caption{Henry~Beck Tube Map (first edition), 1933. Printed at Waterlow \& Sons Ltd., London.}
    \label{fig:beck1933}
\end{figure}

For drawings of (planar) graphs to be readable, special care is
needed to keep the number of bends small. However, the problem of
determining whether a given embedded 8-planar graph (that is, a
planar graph of maximum degree eight with given combinatorial
embedding) admits a bendless octilinear drawing is NP-hard
\cite{Noellenburg05}. This negative result motivated us to study
octilinear drawings of low \emph{edge complexity}, that is, with few
bends per edge. Surprisingly enough, very few results relevant to
this problem were known, even if the octilinear model has been
well-studied in the context of metro-map visualization and map
schematization (see e.g.~\cite{Wolff13}). As an immediate byproduct
of a result of Keszegh et al.~\cite{KPP13}, it turns out that every
$d$-planar graph, with $3 \leq d \leq 8$, admits a planar octilinear
drawing with at most two bends per edge; see
Section~\ref{sec:relatedwork}. On the other hand, every 3-planar
graph on five or more vertices admits a planar octilinear drawing in
which all edges are bendless~\cite{GLM14,Kant92}.

In this paper, we bridge the gap between the two aforementioned
results. In particular, we prove that every 4-planar graph admits a
planar octilinear drawing with at most one bend per edge in cubic
area (see Section~\ref{sec:4planar}). We further show that every
5-planar graph also admits a planar octilinear drawing with at most
one bend per edge, but our construction may require super-polynomial
area (see Section~\ref{sec:5planar}). Hence, we have made no effort
in proving a concrete bound. We complement our results by
demonstrating an infinite class of 6-planar graphs whose planar
octilinear drawings require at least two bends per edge (see
Section~\ref{sec:6planar}).

%=================================================================
\section{Related Work}
\label{sec:relatedwork}
%=================================================================

The research on the (planar) \emph{slope number of graphs} focuses
on minimizing the number of used slopes (see e.g.,
\cite{JJ13,KPP13,KPPT08,LLMN13,MP11}). Octilinear drawings can be
seen as a special case thereof, since only four slopes (horizontal,
vertical and the two diagonals) are used. In a related work, Keszegh
et al.~\cite{KPP13} showed that any $d$-planar graph admits a planar
drawing with one bend per edge, in which all edge-segments have at
most $2d$ different slopes. So, for $d=4$ and for $d=5$, we
significantly reduce the number of different slopes from $8$ and
$10$, resp., to $4$. They also proved that $d$-planar graphs, with
$d \geq 3$, admit planar drawings with two bends per edge that
require at most $\lceil\frac{d}{2}\rceil$ different slopes. It is
not difficult to transfer this technique to the octilinear model and
show that any $d$-planar graph, with $3 \leq d \leq 8$, admits a
planar octilinear drawing with two bends per edge. However, for
$d=3$, Di~Giacomo et al.~\cite{GLM14} recently proved that any
3-planar graph with $n \geq 5$ vertices has a bendless planar
drawing with at most $4$ different slopes and angular resolution
$\pi/4$ (see also~\cite{Kant92}); their approach also yields
octilinear drawings.

Octilinear drawings can be considered as an extension of
\emph{orthogonal drawings}, which allow only horizontal and vertical
segments (i.e., graphs of maximum degree $4$ admit such drawings).
Tamassia~\cite{Tamassia87} showed that one can minimize the total
number of bends in orthogonal drawings of embedded 4-planar graphs.
However, minimizing the number of bends over all embeddings of a
4-planar graph is NP-hard~\cite{GT01}. Note that the core of
Tamassia's approach is a min-cost flow algorithm that first
specifies the angles and the bends of the drawing, producing an
\emph{orthogonal representation}, and then based on this
representation computes the actual drawing by specifying the exact
coordinates for the vertices and the bends of the edges. It is known
that Tamassia's algorithm can be employed to produce a bend-minimum
octilinear representation for any given embedded 8-planar graph.
However, a bend-minimum octilinear representation may not be
realizable by a corresponding planar octilinear drawing \cite{BT04}.
Furthermore, the number of bends on a single edge might be very
high, but can easily be bounded by applying appropriate capacity
constraints to the flow-network.

Biedl and Kant~\cite{BK94} showed that any 4-planar graph except the
octahedron admits a planar orthogonal drawing with at most two bends
per edge on an $O(n^2)$ integer grid. Hence, the octilinear drawing
model allows us to reduce the number of bends per edge at the cost
of an increased area. On the other hand, not all 4-planar graphs
admit orthogonal drawings with one bend per edge; however, testing
whether a 4-planar graph admits such a drawing can be done in
polynomial time~\cite{BKRW14}. In the context of metro-map
visualization, several approaches have been proposed to produce
metro-maps using octilinear or nearly-octilinear polylines, such as
force-driven algorithms~\cite{HMN06}, hill climbing multi-criteria
optimization techniques~\cite{SROW11} and mixed-integer linear
programs~\cite{NW11}. However, the problem of laying out a metro-map
in an octilinear fashion is significantly more difficult than the
octilinear graph drawing problem, as several metro-lines may connect
the same pair of stations and the positions of the vertices have to
reflect geographical coordinates of the stations.

%=================================================================
\section{Preliminaries}
\label{sec:preliminaries}
%=================================================================

In our algorithms, we incrementally construct the drawings similar
to the method of Kant~\cite{Kant92b}. We first employ the canonical
order to cope with triconnected graphs. Then, we extend them to
biconnected graphs using the SPQR-tree~\cite{BT90} and to simply connected
graphs using the BC-tree. In this section we briefly recall them;
however we assume basic familiarity.

\begin{definition}[Canonical order~\cite{Kant92b}]
For a given triconnected plane graph $G=(V,E)$ let $\Pi =
(P_0,\ldots,P_m)$ be a partition of $V$ into paths such that $P_0 =
\{v_1,v_2\}$, $P_m=\{v_n\}$ and $v_2 \rightarrow v_1 \rightarrow
v_n$ is a path on the outer face of $G$. For $k=0,\ldots,m$ let
$G_k$ be the subgraph induced by $\cup_{i=0}^k P_i$ and assume it 
inherits its embedding from $G$. Partition $\Pi$ is a canonical
order of $G$ if for each $k=1,\ldots,m-1$ the following hold:
\begin{inparaenum}[(i)]
\item $G_k$ is biconnected,
\item all neighbors of $P_k$ in $G_{k-1}$ are on the outer face,
of $G_{k-1}$
\item all vertices of $P_k$ have at least one neighbor in $P_j$
for some $j > k$.
\end{inparaenum}
$P_k$ is called a singleton if $|P_k|=1$ and a chain otherwise.
\end{definition}

\begin{definition}[BC-tree]
The \emph{BC-tree} $\mathcal{B}$ of a connected graph $G$ has a
B-node for each biconnected component of $G$ and a C-node for each
cutvertex of $G$. Each B-node is connected with the C-nodes that are
part of its biconnected component.
\end{definition}

An SPQR-tree~\cite{BT90} provides information about the
decomposition of a biconnected graph into its triconnected
components. It can be computed in linear time and space~\cite{GM00}.
Every triconnected component is associated with a node
$\mu$ in the SPQR-tree $\mathcal{T}$. The triconnected component
itself is referred to as the \emph{skeleton} of $\mu$, denoted by
$\skel{\mu} = (\Vskel{\mu},\Eskel{\mu})$. We refer to the degree of
a vertex $v \in \Vskel{\mu}$ in $\skel{\mu}$ as $\sdeg{\mu}{v}$. We
say that $\mu$ is an \emph{R-node}, if $\skel{\mu}$ is a simple
triconnected graph. A bundle of at least three parallel edges
classifies $\mu$ as a \emph{P-node}, while a simple cycle of length
at least three classifies $\mu$ as an \emph{S-node}. By construction
R-nodes are the only nodes of the same type that are allowed to be
adjacent in $\mathcal{T}$. The leaves of $\mathcal{T}$ are formed by
the \emph{Q-nodes}. Their skeleton consists of two parallel edges;
one of them corresponds to an edge of $G$ and is referred to as
\emph{real edge}. The skeleton edges that are not real are referred
to as \emph{virtual edges}. A virtual edge $e$ in $\skel{\mu}$
corresponds to a tree node $\mu'$ that is adjacent to $\mu$ in
$\mathcal{T}$, more exactly, to another virtual edge $e'$ in
$\skel{\mu'}$. We assume that $\mathcal{T}$ is rooted at a Q-node.
Hence, every skeleton (except the one of the root) contains exactly
one virtual edge $e = (s, t)$ that has a counterpart in the skeleton
of the parent node. We call this edge the \emph{reference edge} of
$\mu$ denoted by $\redge{\mu}$. Its endpoints, $s$ and $t$, are
named the \emph{poles} of $\mu$ denoted by $\mathcal{P}_\mu =
\{s,t\}$. Every subtree rooted at a node $\mu$ of $\mathcal{T}$
induces a subgraph of $G$ called the \emph{pertinent graph} of $\mu$
that we denote by $\pert{\mu} = (\Vpert{\mu}, \Epert{\mu})$. We
abbreviate the degree of a node $v$ in $\pert{\mu}$ with
$\pdeg{\mu}{v}$. The pertinent graph is the subgraph of $G$ for
which the subtree describes the decomposition.

Now, assume that $G$ be a simple, biconnected $k$-planar graph,
whose SPQR-tree $\mathcal{T}$ is given. Additionally, we may assume
that $\mathcal{T}$ is rooted at a Q-node that is adjacent to an S-
or R-node. Notice that at least one such node exists since the graph
does not contain any multi-edges, which would be the case if only a
P-node existed. Biconnectivity and maximum degree of $k$ yield basic
bounds for the graph degree of a node $v \in V$, i.e., $2 \leq
\gdeg{v} \leq k$. By construction the pertinent graph of a tree node
$\mu$ is a (connected) subgraph of $G$; thus $1 \leq \pdeg{\mu}{v}
\leq \gdeg{v}$. For the degrees in a skeleton graph $\skel{\mu}$, we
obtain bounds based on the type of the corresponding node. Skeletons
of Q-nodes are cycles of length two, whereas S-nodes are by
definition simple cycles of length at least three; hence,
$\sdeg{\mu}{v} = 2$. For P- and R-nodes the degree can be bounded by
$3 \leq \sdeg{\mu}{v} \leq k$, since the skeleton of the former is
at least a bundle of three parallel virtual edges and the latter's
skeleton is triconnected by definition. The upper bound is derived
from the relation between skeleton and graph degrees: A virtual edge
$e= (s, t)$ hides at least one incident edge of each node (not
necessarily an $(s, t)$-edge). This observation can be easily proven
by induction on the tree. Hence, $2 \leq \sdeg{\mu}{v} \leq
\gdeg{v}$.

Next, we use this observation to derive bounds for the pertinent
degree by distinguishing two cases depending on whether $v$ is a
pole or not. Recall that $\pert{\mu}$ is a subgraph of $G$ that is
obtained by recursively replacing virtual edges by the skeletons of
the corresponding children. In the first case when $v$ is an
internal node in $\pert{\mu}$, i.e., $v \notin \poles{\mu}$, $v$ is
not incident to the reference edge in $\skel{\mu}$. Thus, every edge
of $G$ hidden by a virtual edge in $\skel{\mu}$ is in $\pert{\mu}$.
Hence, $\sdeg{\mu}{v} \leq \pdeg{\mu}{v} \leq k$. In the other case,
i.e., $v \in \mathcal{P}_\mu$, at least one edge that is hidden by
the reference edge, is not part of $\pert{\mu}$, thus,
$\sdeg{\mu}{v} - 1 \leq \pdeg{\mu}{v} \leq k - 1$. Notice that the
lower bounds depend on the skeleton degree which in turn depends on
the type of node, unlike the upper bounds that hold for all tree
nodes. The next lemma tightens these bounds based on the type of the
parent node.

\begin{lemma}
Let $\mu$ be a tree node that is not the root in the SPQR-tree
$\mathcal{T}$ of a simple, biconnected, $k$-planar graph $G$ and
$\mu'$ its parent in $\mathcal{T}$. For $v \in \poles{\mu}$, it
holds that $\pdeg{\mu}{v} \leq k-2$, if $\mu'$ is a P- or an R-node
and $\pdeg{\mu}{v} \leq k-1$ otherwise, i.e. $\mu'$ is an S- or a
Q-node. \label{lem:pdeg_bounds}
\end{lemma}
\begin{proof}
Since the case where $\mu'$ is an S- or a Q-node follows from the
fact that $G$ is k-planar and the reference edge hides at least
one edge that is not in $\pert{\mu}$, we restrict ourselves to the
more interesting cases where $\mu'$ is either a P- or an R-node.
From our previous observations we know that $3 \leq \sdeg{\mu'}{v}
\leq k$. Each of the at least three edges in $\skel{\mu'}$ hides at
least one edge of $G$ that is incident to $v$. However, the total
number of edges is at most $k$ due to the degree restriction. Hence,
we are left with the problem of $k$ edges of $G$ being hidden by at
least three virtual edges, each hiding at least one. As a result the
virtual edge that corresponds to $\mu$ cannot contribute more than
$k-2$ edges to its pertinent graph $\pert{\mu}$.
\end{proof}

\begin{lemma}
In the SPQR-tree $\mathcal{T}$ of a planar biconnected graph $G =
(V,E)$ with $\deg(v) \geq 3$ for every $v \in V$, there exists at
least one Q-node that is adjacent to a P- or an R-node.
\label{lem:pr_node}
\end{lemma}
\begin{proof}
Assume to the contrary that all Q-nodes are adjacent to S-nodes
only. We pick such a Q-node and root $\mathcal{T}$ at it. Let $\mu$
be an S-node (possibly the root itself) with poles $\mathcal{P}_\mu
= \{ s, t \}$ such that there is no other S-node in the subtree of
$\mu$. By definition of an S-node, $\mu$ has at least two children.
If all of them were Q-nodes then there exists a $v \in \Vskel{\mu}$
with $s \neq v \neq t$ and $\deg(v) = 2$; a contradiction. Hence,
there is at least one child $\mu'$ that is a P- or an R-node.
However, since the leaves of $\mathcal{T}$ are Q-nodes and those are
not allowed to be children of P- and R-nodes by our assumption,
there exists at least one other S-node in the subtree of $\mu'$ and
therefore in the subtree of $\mu$ which contradicts our choice of
$\mu$.
\end{proof}

%=================================================================
\section{Octilinear Drawings of 4-Planar Graphs}
\label{sec:4planar}
%=================================================================

In this section, we focus on planar octilinear drawings of 4-planar
graphs. We first consider the case of triconnected 4-planar graphs
and then we extend our approach first to biconnected and then to
simply connected graphs. Central in our approach is the port
assignment; by the \emph{port} of a vertex we refer to the side of
the vertex an edge is connected to. The different ports on a vertex
are distinguished by the cardinal directions.

%=================================================================
\subsection{The Triconnected Case}
\label{sec:4tricon}
%=================================================================

Let $G=(V,E)$ be a triconnected 4-planar graph and $\Pi = \{ P_0,
\ldots, P_m\}$ be a canonical order of $G$. We momentarily neglect
the edge $(v_1,v_2)$ of the first partition $P_0$ of $\Pi$ and we
start by placing the second partition, say a chain $P_1 = \{ v_3,
\ldots, v_{|P_1|+2} \}$, on a horizontal line from left to right.
Since by definition of $\Pi$, $v_3$ and $v_{|P_1|+2}$ are adjacent
to the two vertices, $v_1$ and $v_2$, of the first partition $P_0$,
we place $v_1$ to the left of $v_3$ and $v_2$ to the right of
$v_{|P_1|+2}$. So, they form a single chain where all edges are
drawn using horizontal line-segments that are attached to the east
and west port at their endpoints. The case where $P_1$ is a
singleton is analogous. Having laid out the base of our drawing, we
now place in an incremental manner the remaining partitions. Assume
that we have already constructed a drawing for $G_{k-1}$ and we now
have to place $P_k$, for some $k=2,\ldots,m-1$.

In case where $P_k = \{ v_i, \ldots, v_j \}$ is a chain of $j - i +
1$ vertices, we draw them from left to right along a horizontal line
one unit above $G_{k-1}$. Since $v_i$ and $v_j$ are the only
vertices that are adjacent to vertices in $G_{k-1}$, both only to
one, we place the chain between those two as in
Fig.\ref{fig:4p_chain}. The port used at the endpoints of $P_k$ in
$G_{k-1}$ depends on the following rule: Let $v_{i}'$ ($v_{j}'$,
resp.) be the neighbor of $v_i$ ($v_j$, resp.) in $G_{k-1}$. If the
edge $(v_i, v_i')$ ($(v_j, v_j')$, resp.) is the last to be attached
to vertex $v_i'$ ($v_j'$, resp.), i.e., there is no vertex $v$ in
$P_l \in \Pi$, $l > k$ such that $(v_i',v) \in E$ ($(v_j',v) \in E$,
resp.), then we use the northern port of $v_{i}'$ ($v_{j}'$, resp.).
Otherwise, we choose the north-east port for $(v_i, v_i')$ or the
north-west port for $(v_j, v_j')$.

In case of a singleton $P_k = \{ v_i \}$, we can apply the previous
rule if the singleton is of degree three, as the third neighbor of
$v_i$ should belong to a partition $\Pi_j$ for some $j > k$.
However, in case where $v_i$ is of degree four we may have to deal with
an additional third edge $(v_i, v)$ that connects $v_i$ with
$G_{k-1}$. By the placement so far, we may assume that $v$ lies
between the other two endpoints, thus, we place $v_i$ such that
$x(v_i) = x(v)$. This enables us to draw $(v_i, v)$ as a vertical
line-segment; see Fig.\ref{fig:4p_singleton}.

The above procedure is able to handle all chains and singletons
except the last partition $P_m$, because $v_n$ may have $4$ edges
pointing downwards. One of these edges is $(v_n, v_1)$, by
definition of $\Pi$. We exclude $(v_n, v_1)$ and draw $v_n$ as an
ordinary singleton. Then, we shift $v_1$ to the left and up as in
Fig.\ref{fig:4p_final}. This enables us to draw $(v_1, v_n)$ as a
horizontal-vertical segment combination. For $(v_1, v_2)$, we move
$v_2$ one unit to the right and down. We free the west port of $v_2$
by redrawing its incident edges as in Fig.\ref{fig:4p_final} and
attach $(v_1, v_2)$ to it. Edge $(v_1,v_2)$ will be drawn as a
diagonal segment with positive slope connected to $v_1$ and a
horizontal segment connected to $v_2$, which requires one bend. Let
$(v_2,v_i)$ be the other incomplete edge according to Figure
\ref{fig:4p_final}. It will be drawn using a diagonal segment with
positive slope connected to $v_2$ and a horizontal segment connected
to $v_i$, again requiring one bend.

So far, we have specified a valid port assignment and the
y-coordinates of the vertices. However, we have not fully specified
their x-coordinates. Notice that by construction every edge, except
the ones drawn as vertical line-segments, contains exactly one
horizontal segment. This enables us to stretch the drawing
horizontally by employing appropriate cuts. A \emph{cut}, for us, is
a $y$-monotone continuous curve that crosses only horizontal
segments and divides the current drawing into a left and a right
part. It is not difficult to see that we can shift the right part of
the drawing that is defined by the cut further to the right while
keeping the left part of the drawing on place and the result remains
a valid octilinear drawing.

To compute the x-coordinates, we proceed as follows. We first assign
consecutive x-coordinates to the first two partitions and from there
on we may have to stretch the drawing in two cases. The first one
appears when we introduce a chain, say $P_k$, as it may not fit into
the gap defined by its two adjacent vertices in $G_{k-1}$. In this
case, we horizontally stretch the drawing between its two adjacent
vertices in $G_{k-1}$ to ensure that their horizontal distance is at
least $|P_{k}| + 1$. The other case appears when an edge that
contains a diagonal segment is to be drawn. Such an edge requires a
horizontal distance between its endpoints that is at least the
height it bridges. We also have to prevent it from intersecting any
horizontal-vertical combinations in the face below it. We can cope
with both cases by horizontally stretching the drawing by a factor
that is bounded by the current height of the drawing. Since the
height of the resulting drawing is bounded by $|\Pi|=O(n)$, it
follows that in the worst case its width is $O(n^2)$. We are now
ready to state the main theorem of this subsection.

\begin{figure}[t]
    \centering
    \begin{minipage}[b]{.32\textwidth}
        \centering
        \subfloat[\label{fig:4p_chain}{}]
        {\includegraphics[width=.9\textwidth]{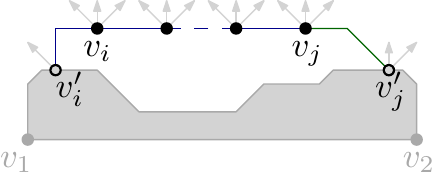}}
    \end{minipage}
    \begin{minipage}[b]{.32\textwidth}
        \centering
        \subfloat[\label{fig:4p_singleton}{}]
        {\includegraphics[width=.9\textwidth]{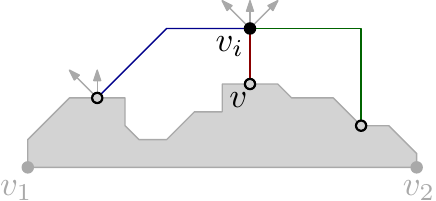}}
    \end{minipage}
    \begin{minipage}[b]{.32\textwidth}
        \centering
        \subfloat[\label{fig:4p_final}{}]
        {\includegraphics[width=\textwidth,page=3]{4p_final}}
    \end{minipage}
    \caption{
    (a)~Horizontal placement of a chain $P_k = \{ v_i, \ldots, v_j\}$.
    (b)~Placement of a singleton $P_k = \{ v_i \}$ with degree four.
    (c)~Final layout after repositioning $v_1$ and $v_2$ (the shape of the dotted edges can be obtained by extending the stubs until they intersect).}
    \label{fig:4p_canonical}
\end{figure}

\begin{theorem}
Given a triconnected 4-planar graph $G$, we can compute in $O(n)$
time an octilinear drawing of $G$ with at most one bend per edge on
an $O(n^2) \times O(n)$ integer grid.
\end{theorem}
\begin{proof}
In order to keep the time complexity of our algorithm linear, we
employ a simple trick. We assume that any two adjacent points of the
underlying integer grid are by $n$ units apart in the horizontal
direction and by one unit in the vertical direction. This a priori
ensures that all edges that contain a diagonal segment will not be
involved in crossings and simultaneously does not affect the total
area of the drawing, which asymptotically remains cubic. On the
other hand, the advantage of this approach is that we can use the
shifting method of Kant~\cite{Kant92b} to cope with the introduction
of chains in the drawing, that needs $O(n)$ time in total by keeping
relative coordinates that can be efficiently updated and computing
the absolute values only at the last step.
\end{proof}

Note that our algorithm produces drawings that have a linear number
of bends in total (in particular, exactly $2|\Pi|=O(n)$ bends). In
the following, we prove that this bound is asymptotically tight.

\begin{theorem}
There exists an infinite class of 4-planar graphs which do not admit
bendless octilinear drawings and if they are drawn with at most one
bend per edge, then a linear number of bends is required.
\end{theorem}
\begin{proof}
Based on the simple fact that in an orthogonal drawing a triangle
requires at least one bend, we describe an example that translates
this idea to the octilinear model (see Fig.\ref{fig:4p_lowerbound}).
While a triangle can easily be drawn bendless with the additional
ports available, we will occupy those to enforce the creation of a
bend as in the orthogonal model. Furthermore, the example is
triconnected. Hence, its embedding is fixed up to the choice of the
outer face. Our construction is heavily based on the so called
\emph{separating triangle}, i.e., a three-cycle whose removal
disconnects the graph. Each vertex of such a triangle has degree
four. Any triangle which is drawn bendless has a $45^\circ$ angle
inside. But since the triangles are nested and have incident edges
going inside of the triangles, this is impossible.
\end{proof}

\begin{figure}[t]
    \centering
    \includegraphics[width=.4\textwidth]{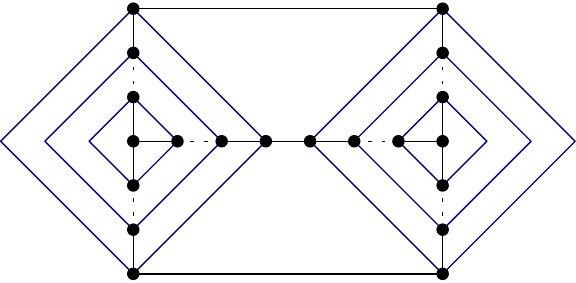}
    \caption{Nested separating triangles each requiring one bend.}
    \label{fig:4p_lowerbound}
\end{figure}

%=================================================================
\subsection{The Biconnected Case}
\label{sec:4bicon}
%=================================================================

Following standard practice, we employ a rooted SPQR-tree and assume
for a tree node that the pertinent graphs of its children are drawn
in a pre-specified way. Consider a node $\mu$ in $\mathcal{T}$ with
poles $\mathcal{P}_\mu = \{ s, t \}$. In the drawing of
$\pert{\mu}$, $s$ should be located at the upper-left and $t$ at the
lower-right corner of the drawing's bounding box with a port
assignment as in Fig.\ref{fig:4planar_layout_spec}. In general, we
assume that the edges incident to $s$ ($t$, resp.) use the western
(eastern, resp.) port at their other endpoint, except of the
northern (southern, resp.) most edge which may use the north (south,
resp.) port instead. In that case we refer to $s$ and $t$ as
\emph{fixed}; see $\overline{e}_s, \overline{e}_t$ in
Fig.\ref{fig:4planar_layout_spec}. More specifically, we maintain
the following invariants:

\begin{enumerate}[{I}P-1:]
\item \label{ip:1} The width (height) of the drawing of $\mu$ is
quadratic (linear) in the size of $\pert{\mu}$. $s$ is located at
the upper-left; $t$ at the lower-right corner of the drawing's
bounding box.
\item \label{ip:2} If $\pdeg{\mu}{s} \geq 2$, $s$ is fixed;
$t$ is fixed if $\pdeg{\mu}{t} = 3$ and $\mu$'s parent is not the
root.
\item \label{ip:3} The edges that are incident at $s$ and $t$ in
$\pert{\mu}$ use the south, south-east and east ports at $s$ and the
north, north-west and west port at $t$, resp. If $s$ or $t$ is not
fixed, incident edges are attached at their other endpoints via the
west and east port, respectively. If $s$ or $t$ is fixed, the
northern-most edge at $s$ and the southern-most edge at $t$ may use
the north (south, resp.) port at its other endpoint.
\end{enumerate}

Notice that the port assignment, i.e. \IPPort, guarantees the
ability to stretch the drawing horizontally even in the case where
both poles are fixed. Furthermore, \IPFix is \emph{interchangeable}
in the following sense: If $\pdeg{\mu}{s} = 2$ and $\pdeg{\mu}{t} =
1$, then $s$ is fixed but $t$ is not. But, if we relabel $s$ and $t$
such that $t'= s$ and $s' = t$, then $\pdeg{\mu}{s'} = 1$ and
$\pdeg{\mu}{t'} = 2$. By \IPFix, we can create a drawing where both
$s'$ and $t'$ are not fixed and located in the upper-left and
lower-right corner of the drawing's bounding box. Afterwards, we
mirror the resulting layout vertically and horizontally to obtain
one where $s$ and $t$ are in their respective corners and not fixed.
Notice that in general the property of being fixed is not symmetric,
e.g., when $\pdeg{\mu}{s} = 3$ and $\pdeg{\mu}{t} = 2$ holds, $s$
remains fixed while $t$ becomes fixed as well. For a non-fixed
vertex, we introduce an operation that is referred to as forming or
creating a \emph{nose}; see Fig.\ref{fig:4planar_nose_example},
where $t$ has been moved downwards at the cost of a bend. As a
result, the west port of $t$ is no longer occupied.

\begin{figure}[t]
    \centering
    \begin{minipage}[b]{.24\textwidth}
        \centering
        \subfloat[\label{fig:4planar_layout_spec}{}]
        {\includegraphics[width=.95\textwidth]{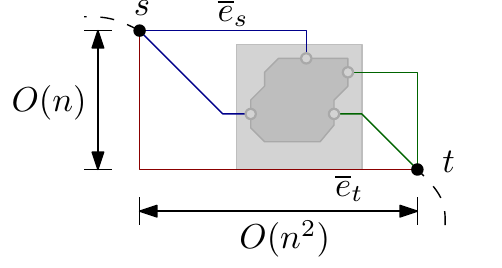}}
    \end{minipage}
    \hfill
    \begin{minipage}[b]{.24\textwidth}
        \centering
        \subfloat[\label{fig:4planar_nose_example}{}]
        {\includegraphics[width=.95\textwidth,page=1]{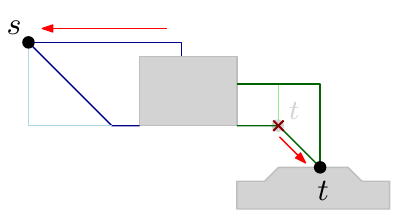}}
    \end{minipage}
    \hfill
    \begin{minipage}[b]{.24\textwidth}
        \centering
        \subfloat[\label{fig:4p_P_case_1}{}]
        {\includegraphics[width=.95\textwidth,page=1]{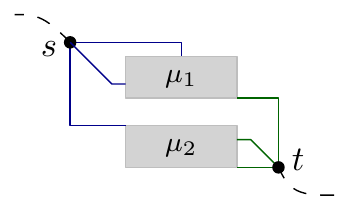}}
    \end{minipage}
    \hfill
    \begin{minipage}[b]{.24\textwidth}
        \centering
        \subfloat[\label{fig:4p_P_case_2}{}]
        {\includegraphics[width=.95\textwidth,page=2]{4p_P_case}}
    \end{minipage}
    \caption{
    (a)~Schematic view of the layout requirements.
    (b)~Creating a nose at $t$.
    (c)~First P-node subcase without an $(s,t)$-edge but $s$ might be fixed in a child $\mu_1$.
    (d)~Second P-node subcase with an $(s,t)$-edge where $t$ might get fixed in a child $\mu_2$.}
    \label{fig:4p_requirements}
\end{figure}

\begin{description}
\item[P-node case:] Let $\mu$ be a P-node. By Lemma~\ref{lem:pdeg_bounds},
for a child $\mu'$ of $\mu$, it holds that $\pdeg{\mu'}{s} \leq 2$
and $\pdeg{\mu'}{t} \leq 2$. So, $t$ can form a nose in $\mu'$,
while $s$ might be fixed in the case where $\pdeg{\mu'}{s} = 2$.
Notice that there exists at most one such child due to the degree
restriction. We distinguish two cases based on the existence of an
$(s,t)$-edge.

In the first case, assume that there is no $(s,t)$-edge, i.e., there
is no child that is a Q-node. We draw the children of $\mu$ from top
to bottom such that a possible child in which $s$ is fixed, is drawn
topmost (see $\mu_1$ in Fig.\ref{fig:4p_P_case_1}). In the second
case, we draw the $(s,t)$-edge at the top and afterwards the
remaining children (see Fig.\ref{fig:4p_P_case_2}). Of course, this
works only if $s$ is not fixed in any of the other children. Let
$\mu'$ be such a potential child where $s$ is fixed, i.e.,
$\pdeg{\mu'}{s} = 2$, and thus, the only child that remains to be
drawn. Here, we use the property of interchangeability to ``unfix"
$s$ in $\mu'$. As a result $s$ can form a nose, whereas $t$ may now
be fixed in $\mu'$ when $\pdeg{\mu'}{t} = 2$ holds, as in
Fig.\ref{fig:4p_P_case_2}. However, then $\pdeg{\mu}{t} = 3$
follows. Notice that the presence of an $(s,t)$-edge implies that
the parent of $\mu$ is not the root of $\mathcal{T}$, since this
would induce a pair of parallel edges. Hence, by \IPFix we are
allowed to fix $t$ in $\mu$. Port assignment and area requirements
comply in both cases with our invariant properties.
\item[S-node case:] We place the drawings of the children,
say $\mu_1, \ldots, \mu_\ell$, of an S-node $\mu$ in a ``diagonal
manner'' such that their corners touch as in
Fig.\ref{fig:4p_S_case}. In case of Q-nodes being involved, we draw
their edges as horizontal segments (see, e.g., edge $(v_3,v_4)$ in
Fig.\ref{fig:4p_S_case} that corresponds to Q-node $\mu_3$). Observe
that $s$ and $t$ inherit their port assignment and pertinent degree
from $\mu_1$ and $\mu_\ell$, respectively, i.e., $\pdeg{\mu}{s} =
\pdeg{\mu_1}{s}$ and $\pdeg{\mu}{t} = \pdeg{\mu_\ell}{t}$. So, we
may assume that $s$ is fixed in $\mu$, if $s$ is fixed in $\mu_1$.
Similarly, $t$ is fixed in $\mu$, if $t$ is fixed in $\mu_\ell$. By
\IPFix, $t$ is not allowed to be fixed in the case where the parent
of $\mu$ is the root of $\mathcal{T}$. However, Lemma
\ref{lem:pr_node} states that we can choose the root such that $t$
is not fixed in that case, and thus, complies with \IPFix. Since we
only concatenated the drawings of the children, \IPGeo and \IPPort
are satisfied.
\item[R-node case:] For the case where $\mu$ is an R-node with poles
$\mathcal{P}_{\mu} = \{s, t\}$, we follow the basic idea of the
triconnected algorithm of the previous section and describe the
modifications necessary to handle the drawing of the children of
$\mu$. To do so, we assume the worst case where no child of $\mu$ is
a Q-node. Let $\mu_{uv}$ denote the child that is represented by the
virtual edge $(u,v) \in \Eskel{\mu}$. Notice that due to
Lemma~\ref{lem:pdeg_bounds}, $\pdeg{\mu_{uv}}{u} \leq 2$ and
$\pdeg{\mu_{uv}}{v} \leq 2$ holds. Hence, with \IPFix we may assume
that at most one out of $u$ and $v$ is fixed in $\mu_{uv}$. We
choose the first partition in the canonical ordering to be $P_0 =
\{s, t\}$ and distinguish again between whether the partition to be
placed next is a chain or a singleton.

\begin{figure}[t]
    \centering
    \begin{minipage}[b]{.36\textwidth}
        \centering
        \subfloat[\label{fig:4p_S_case}{}]
        {\includegraphics[width=\linewidth, page=1]{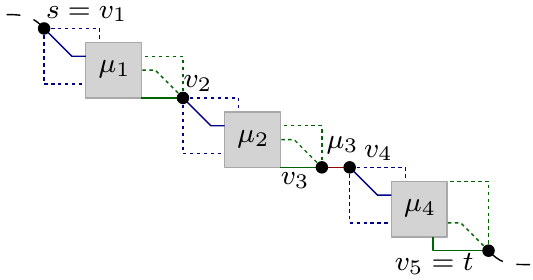}}
    \end{minipage}\hspace{.13\textwidth}
    \begin{minipage}[b]{.36\textwidth}
        \centering
        \subfloat[\label{fig:4p_R_chain_case}{}]
        {\includegraphics[width=\linewidth, page=1]{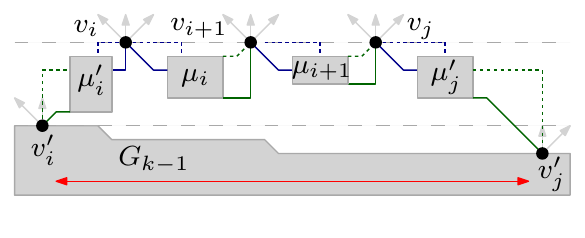}}
    \end{minipage}\\
    \begin{minipage}[b]{.36\textwidth}
        \centering
        \subfloat[\label{fig:4p_R_singleton_case}{}]
        {\includegraphics[width=\linewidth, page=2]{4p_R_case_new}}
    \end{minipage}\hspace{.13\textwidth}
    \begin{minipage}[b]{.36\textwidth}
        \centering
        \subfloat[\label{fig:4p_R_last_vertex_case}{}]
        {\includegraphics[width=\linewidth, page=4]{4p_R_case_new}}
    \end{minipage}
    \caption{
    (a)~S-node with children $\mu_1, \ldots, \mu_4$; $\mu_3$ is a Q-node representing the edge $(v_3, v_4)$. Optional edges are drawn dotted.
    (b)~Example for a chain $v_i, \ldots, v_j$ with virtual edges representing $\mu_i, \ldots, \mu_{j-1}$ in the R-node case.
    (c)~Singleton $v_i$ with possibly three incident virtual edges representing $\mu'_i, \mu'_v, \mu'_j$.
    (d)~Placing $v_n$ and moving up $s$ which might be fixed in $\mu_{sn}$.}
\end{figure}

In case of a chain, say $P_k = \{ v_i, \ldots, v_j \}$ with two
neighbors $v_{i}'$ and $v_{j}'$ in $G_{k-1}$, we have to replace two
types of edges with the drawings of the corresponding children: the
edges $(v_i, v_{i+1}), \ldots, (v_{j-1}, v_{j})$ representing the
children $\mu_i, \ldots, \mu_{j-1}$ and $(v_{i}', v_{i})$ ($(v_{j},
v_{j}')$ resp.) representing $\mu_{i}'$ ($\mu_{j}'$ resp.). We place
the vertices of $P_k$ on a horizontal line high enough above
$G_{k-1}$ such that every drawing may fit in-between it and
$G_{k-1}$. Then, we insert the drawings aligned below the horizontal
line and choose for $i \leq l < j$, $v_{l}$ to be the fixed node in
$\mu_{l}$, whereas in $\mu_{i}'$ ($\mu_{j}'$ resp.), we set $v_{i}$
($v_j$ resp.) to be fixed. Hence, for $i \leq l < j$, $v_{l+1}$ may
form a nose in $\mu_{l}$ pointing upwards while $v_{i}'$ and
$v_{j}'$ form each one downwards as depicted in
Fig.\ref{fig:4p_R_chain_case}. For the extra height and width, we
stretch the drawing horizontally.

For the case where $P_k = \{ v_i \}$ and $i \neq n$ is a singleton,
we only outline the difference which is a possible third edge $(v_i,
v)$ to $G_{k-1}$ representing say $\mu_{v}'$. While the other two
involved children, say  $\mu_{i}'$ and $\mu_{j}'$, are handled as in
the chain-case, $\mu_{v}'$ requires extra height now and we may
place $v_i$ such that $\mu_{v}'$ fits below $\mu_{j}'$ as in
Fig.\ref{fig:4p_R_singleton_case}. Notice that $\pdeg{\mu_{v}'}{v_i}
= 1$ holds and therefore by \IPFix both $v_i$ and $v$ are not fixed
in $\mu_{v}'$. Hence, forming a nose at $v_i$ and $v$ as in
Fig.\ref{fig:4p_R_singleton_case} is feasible.

It remains to describe the special case where the last singleton
$P_k = \{ v_n \}$ is placed. Since $s,t \in P_0$, both have not been
fixed yet. We proceed as in the triconnected algorithm and move $s =
v_1$ above $v_n$ as depicted in Fig.\ref{fig:4p_R_last_vertex_case},
high enough to accommodate the drawing of the child $\mu_{sn}$
represented by the edge $(s, v_n)$. Since we may require $v_n$ to
form a nose in $\mu_{sn}$ as in Fig.\ref{fig:4p_R_last_vertex_case},
we choose $s$ to be fixed in $\mu_{sn}$. However, we are allowed by
\IPFix to fix $s$ since $t$ remains unfixed. For the area
constraints of \IPGeo, we argue as follows: Although some diagonal
segments may force us to stretch the whole drawing by its height,
the height of the drawing has been kept linear in the size of
$\pert{\mu}$. Since we increase the width by the height a constant
number of times per step, the resulting width remains quadratic.
\item[Root case:] For the root of $\mathcal{T}$ we distinguish two cases:
In the first case, there exists a vertex $v \in V$ with $\deg(v)
\leq 3$. Then, we choose as root a Q-node $\mu$ that represents one
of its three incident edges and orient the poles $\{s,t\}$ such that
$t = v$. Hence, for the child $\mu'$ of $\mu$ follows
$\pdeg{\mu'}{t} \leq 2$. In the other case, i.e., for every $v \in
V$ we have $\deg(v) = 4$, we choose a Q-node that is not adjacent to
an S-node, whose existence is guaranteed by Lemma~\ref{lem:pr_node}.
In both cases, we may form a nose with $t$ pointing downwards and
draw the edge as in the triconnected algorithm.
\end{description}

\begin{theorem}
Given a biconnected 4-planar graph $G$, we can compute in $O(n)$
time an octilinear drawing of $G$ with at most one bend per edge on
an $O(n^2) \times O(n)$ integer grid.
\end{theorem}
\begin{proof}
The SPQR-tree $\mathcal{T}$ can be computed in $O(n)$-time and its
size is linear to the size of $G$ \cite{GM00}. The pertinent degrees
of the poles at every node can be pre-computed by a bottom-up
traversal of $\mathcal{T}$. Drawing a P-node requires constant time;
S- and R-nodes require time linear to the size of the skeleton.
However, the sum over all skeleton edges is linear, as every virtual
edge corresponds to a tree node.
\end{proof}

%=================================================================
\subsection{The Simply Connected Case}
\label{sec:4con}
%=================================================================
After having shown that we can cope with biconnected 4-planar
graphs, we turn our attention to the connected case. We start by
computing the BC-tree of $G$ and root it at some arbitrary B-node.
Every B-node, except the root, contains a designated cut vertex that
links it to the parent. A \emph{bridge} for a biconnected component
consists only of a single edge. Similar to the biconnected case, we
define an invariant for the drawing of a subtree: The cut vertex
that links the subtree to the parent is located in the upper left
corner of the drawing's bounding box.

\begin{figure}[t]
    \centering
    \begin{minipage}[b]{.24\textwidth}
        \centering
        \subfloat[\label{fig:4p_bc_root}{}]
        {\includegraphics[page=3]{4p_bc}}
    \end{minipage}
    \begin{minipage}[b]{.24\textwidth}
        \centering
        \subfloat[\label{fig:4p_bc_Snode}{}]
        {\includegraphics[page=4]{4p_bc}}
    \end{minipage}
    \begin{minipage}[b]{.24\textwidth}
        \centering
        \subfloat[\label{fig:4p_bc_Rnode}{}]
        {\includegraphics[page=5]{4p_bc}}
    \end{minipage}
    \begin{minipage}[b]{.24\textwidth}
        \centering
        \subfloat[\label{fig:4p_bc_bridges}{}]
        {\includegraphics[page=1]{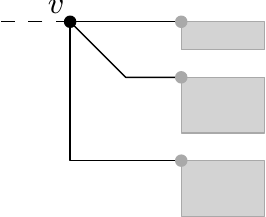}}
    \end{minipage}
    \caption{
    (a)~Rooting the SPQR-tree such that $v_b$ is in the upper-left corner.
    (b)~All possible situations at an S-node $\mu$. For attaching $b_2$ to $v_2$, the layout had to be modified.
    (c)~Attaching a subtree via a bridge to a cut vertex $v_c$ in an R-node. The dashed edge $(v_i, v')$ may only be present if $v_i = v_n$.
    (d)~A cut vertex where all of its children are attached via bridges.}
\end{figure}

Any subgraph, say $G_b$, induced by a non-bridge biconnected
component can be laid out using the biconnected algorithm. However,
to construct a drawing that satisfies our invariant we have to take
care of two problems. First, the cut vertex, say $v_b$, that links
$G_b$ to the parent, has to be drawn in the upper-left corner of the
subtrees drawing. Second, there may be other cut vertices of $G$ in
$G_b$ to which we have to attach their corresponding subtrees.

For the first problem we describe how to root the SPQR-tree
$\mathcal{T}_b$ for $G_b$ so that $v_b$ is located in the upper-left
corner. There are at least two Q-nodes having $v_b$ as a pole (as
$G_b$ is biconnected) and the degree of $v_b$ in $G_b$ is at most
$3$. In the biconnected case, we distinguished for the root of the
tree between whether there exists $v \in V$ with $\deg(v) \leq 3$ or
not. Hence, we may choose for the root of $\mathcal{T}_b$ a Q-node
having $v_b$ as a pole and orient it such that $v_b = t$, thus,
satisfying $\deg(t) \leq 3$. Then, we flip the final drawing of
$G_b$ such that $t$ is in the upper left corner (see
Fig.\ref{fig:4p_bc_root}).

Next, we address the second problem. Let $v_c$ be a cut vertex in
$G_b$ that is not the link to the parent. If $v_c$ has degree $3$,
then it may occur in the pertinent graph of every node. However, in
this case we only have to attach a subtree of the BC-tree that is
connected via a bridge. This poses no problem, as there are enough
free ports available at $v_c$ and we can afford a bend at the
bridge. We only consider S- and R- nodes here since the poles of
P-nodes occur in the pertinent graphs of the first two. For R-nodes
we assume that the south east port at $v_c$ is free. So, we attach
the drawing via the bridge by creating a bend as in
Fig.\ref{fig:4p_bc_Rnode}. In the diagonal drawing of an S-node, the
north-east port is free. So, we can proceed similar; see
Fig.\ref{fig:4p_bc_Snode}.

If $v_c$ has degree $2$ in $G_b$, it only occurs in the pertinent
graph of an S-node; see $v_3$ in Fig.\ref{fig:4p_bc_Snode}. However,
we may no longer assume that the bridge is available. As a result,
we cannot afford a bend and have to deal with two incident edges
instead of one. We modify the drawing by exploiting the two real
edges incident to $v_c$ in the S-nodes layout to free the east and
south east port; see $v_2$ in Fig.\ref{fig:4p_bc_Snode}. This
enables us to attach the subtrees drawing without modifying it. We
finish this section by dealing with the most simple case where there
are only bridges attached to a cut vertex. The idea is illustrated
in Fig.\ref{fig:4p_bc_bridges} and matches our layout specification.

\begin{theorem}
Given a connected 4-planar graph $G$, we can compute in $O(n)$ time
an octilinear drawing of $G$ with at most one bend per edge on an
$O(n^2) \times O(n)$ integer grid. \label{thm:4planarconnected}
\end{theorem}
\begin{proof}
Decomposing a connected graph into its biconnected components takes
linear time. It remains the area property. Inserting a subtree with
$n$ vertices and the given dimensions into the drawing of an R- or
S-node clearly increases the width of the drawing by at most
$O(n^2)$ and the height by at most $O(n)$. Hence, the total drawing
area is cubic, as desired.
\end{proof}

%=================================================================
\section{Octilinear Drawings of 5-Planar Graphs}
\label{sec:5planar}
%=================================================================

In this section, we focus on planar octilinear drawings of 5-planar
graphs. As in Section~\ref{sec:4planar}, we first consider the case
of triconnected 5-planar graphs and then we extend our approach
first to biconnected and then to the simply connected graphs.

%=================================================================
\subsection{The Triconnected Case}
\label{sec:5tricon}
%=================================================================

Let $G=(V,E)$ be a triconnected 5-planar graph and $\Pi = \{ P_0,
\ldots, P_m\}$ be a canonical order of $G$. We place the first two
partitions $P_0$ and $P_1$ of $\Pi$, similar to the case of 4-planar
graphs. Again, we assume that we have already constructed a drawing
for $G_{k-1}$ and now we have to place $P_k$, for some
$k=2,\ldots,m-1$. We further assume that the $x$- and
$y$-coordinates are computed simultaneously so that the drawing of
$G_{k-1}$ is planar and horizontally stretchable in the following
sense: If $e \in E(G_{k-1})$ is an edge incident to the outer face
of $G_{k-1}$, then there is always a cut which crosses $e$ and can
be utilized to horizontally stretch the drawing of $G_{k-1}$. This
is guaranteed by our construction which makes sure that in each step
the edges incident to the outer face have a horizontal segment. In
other words, one can define a cut through every edge incident to the
outer face of $G_{k-1}$ (\emph{stretchability-invariant}).

\begin{figure}[t]
    \centering
    \begin{minipage}[b]{.32\textwidth}
        \centering
        \subfloat[\label{fig:5p_chain}{}]
        {\includegraphics[width=.9\textwidth]{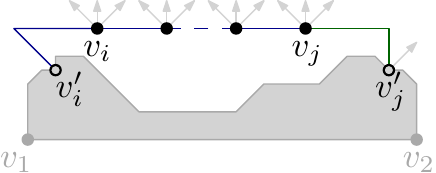}}
    \end{minipage}
    \begin{minipage}[b]{.32\textwidth}
        \centering
        \subfloat[\label{fig:5p_singleton}{}]
        {\includegraphics[width=.9\textwidth]{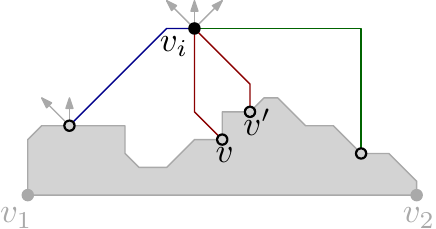}}
    \end{minipage}
    \begin{minipage}[b]{.32\textwidth}
        \centering
        \subfloat[\label{fig:5p_before_final}{}]
        {\includegraphics[width=.9\textwidth,page=1]{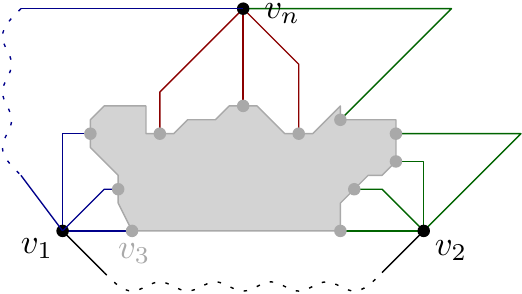}}
    \end{minipage}
    \caption{
    (a)~Horizontal placement of a chain $P_k = \{ v_i, \ldots, v_j\}$.
    (b)~Placement of a singleton $P_k = \{ v_i \}$ of degree five.
    (c)~Final layout (the shape of the dotted edges can be obtained by extending the stubs until they intersect).}
    \label{fig:4p_canonical}
\end{figure}

If $P_k = \left\{v_i,\ldots,v_j\right\}$ is a chain, it is placed
exactly as in the case of 4-planar graphs, but with different port
assignment. Recall that by $v_{i}'$ ($v_{j}'$, resp.) we denote the
neighbor of $v_i$ ($v_j$, resp.) in $G_{k-1}$. Among the northern
available ports of vertex $v_i'$ ($v_{j}'$, resp.), edge $(v_i,
v_i')$ ($(v_j, v_j')$, resp.) uses the eastern-most unoccupied port
of $v_{i}'$ (western-most unoccupied port of $v_{j}'$, resp.); see
Fig.\ref{fig:5p_chain}. If $P_k$ does not fit into the gap between
its two adjacent vertices $v_{i}'$ and $v_{j}'$ in $G_{k-1}$, then
we horizontally stretch $G_{k-1}$ between $v_{i}'$ and $v_{j}'$ to
ensure that the horizontal distance between $v_{i}'$ and $v_{j}'$ is
at least $|P_{k}| + 1$. This can always be accomplished due to the
stretchability-invariant, as both $v_{i}'$ and $v_{j}'$ are on the
outer face of $G_{k-1}$. Potential crossings introduced by edges of
$P_k$ containing diagonal segments can be eliminated by employing
similar cuts to the ones presented in the case of 4-planar graphs.
So, we may assume that $G_{k}$ is plane. Also, $G_{k}$ complies with
the stretchability-invariant, as one can define a cut that crosses
any of the newly inserted edges of $P_k$ and then follows one of the
cuts of $G_{k-1}$ that crosses an edge between $v_{i}'$ and
$v_{j}'$.

In case of a singleton $P_k = \{ v_i \}$ of degree $3$ or $4$, our
approach is very similar to the one of the case of 4-planar graphs.
Here, we mostly focus on the case where $v_i$ is of degree five. In
this case, we have to deal with two additional edges (called
\emph{nested}) that connect $v_i$ with $G_{k-1}$, say $(v_i, v)$ and
$(v_i,v')$; see Fig.\ref{fig:5p_singleton}. Such a pair of edges
does not always allow vertex $v_i$ to be placed along the next
available horizontal grid line; $v_i$'s position is more or less
prescribed, as each of $v$ and $v'$ may have only one northern  port
unoccupied. However, a careful case analysis on the type of ports
(i.e., north-west, north or north-east) that are unoccupied at $v$
and $v'$ in conjunction with the fact that $G_{k-1}$ is horizontally
stretchable shows that we can always find a feasible placement for
$v_i$ (usually far apart from $G_{k-1}$). Potential crossings due to
the remaining edges incident to $v_i$ are eliminated by employing
similar cuts to the ones presented in the case of 4-planar graphs.
So, we may assume that $G_{k}$ is planar. Similar to the case of a
chain, we prove that $G_{k}$ complies with the
stretchability-invariant. In this case special attention should be
paid to avoid crossings with the nested edges of $v_i$, as a nested
edge may contain no horizontal segment. Note that the case of the
last partition $P_m=\{v_n\}$ is treated in the same way, even if
$v_n$ is potentially incident to three nested edges; see
Fig.\ref{fig:5p_before_final}.

To complete the description of our approach it remains to describe
how edge $(v_1,v_2)$ is drawn. By construction both $v_1$ and $v_2$
are along a common horizontal line. So, $(v_1,v_2)$ can be drawn
using two diagonal segments that form a bend pointing downwards; see
Fig.\ref{fig:5p_before_final}.

\begin{theorem}
Given a triconnected 5-planar graph $G$, we can compute in $O(n^2)$
time an octilinear drawing of $G$ with at most one bend per edge.
\end{theorem}
\begin{proof}
Unfortunately, we can no longer use the shifting method of
Kant~\cite{Kant92b}, since the $x$- and $y$-coordinates are not
independent. However, the computation of each cut can be done in
linear time, which implies that our drawing algorithm needs $O(n^2)$
time in total.
\end{proof}

Recall that when placing a singleton $P_k = \left\{v_i\right\}$ that
has four edges to $G_{k-1}$, the height of $G_k$ is determined by
the horizontal distance of its neighbors along the outer face of
$G_{k-1}$, which is bounded by the actual width of the drawing of
$G_{k-1}$. On the other hand, when placing a chain $P_k$ the amount
of horizontal stretching required in order to avoid potential
crossings is delimited by the height of the drawing of $G_{k-1}$.
Unfortunately, this connection implies that for some input
triconnected 5-planar graphs our drawing algorithm may result in
drawings of super-polynomial area, as the following theorem
suggests.

\begin{figure}[t]
    \centering
    \includegraphics[width=.3\textwidth]{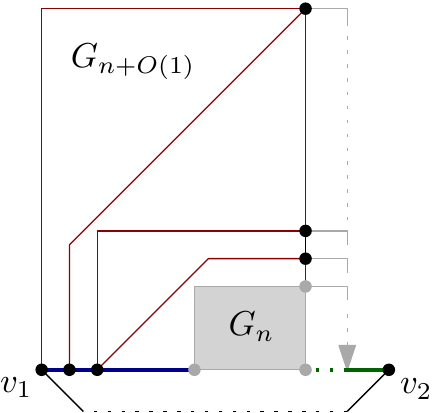}
    \caption{A recursive construction of an infinite class of 5-planar graphs requiring super-polynomial drawing area.}
    \label{fig:5p_exparea}
\end{figure}

\begin{theorem}
There exist infinitely many triconnected 5-planar graphs for which
our drawing algorithm produces drawings of super-polynomial area.
\label{thm:5planarExp}
\end{theorem}
\begin{proof}
Fig.\ref{fig:5p_exparea} illustrates a recursive construction of an
infinite class of 5-planar triconnected graphs with this property.
The base of the construction is a ``long chain'' connecting $v_1$
and $v_2$ (refer to the bold drawn edges of
Fig.\ref{fig:5p_exparea}). Each next member, say $G_{n+O(1)}$, of
this class is constructed by adding a constant number of vertices
(colored black in Fig.\ref{fig:5p_exparea}) to its immediate
predecessor member, say $G_{n}$, of this class, as illustrated in
Fig.\ref{fig:5p_exparea}. If $W_n$ and $H_n$ is the width and the
height of $G_n$, respectively, then it is not difficult to show that
$W_{n+O(1)}>2W_{n}$ and $H_{n+O(1)}>2H_{n}$, which implies that the
required area is asymptotically exponential.
\end{proof}

%=================================================================
\subsection{The Biconnected Case}
\label{sec:bicon}
%=================================================================

For the 4-planar case we defined several invariants in order to keep
the area of the resulting drawings polynomial. Since we drop this
requirement now we can define a (simpler) new invariant for the
biconnected 5-planar case. When considering a node $\mu$ in
$\mathcal{T}$ and its poles $\mathcal{P}_\mu = \{ s, t \}$, then in
the drawing of $\pert{\mu}$, $s$ and $t$ are horizontally aligned at
the bottom of the drawing's bounding box as in
Fig.\ref{fig:5p_layout}. If an $(s,t)$-edge is present, it can be
drawn at the bottom. An $(s,t)$-edge only occurs in the pertinent
graph of a P-node (and Q-node). Again, we use the term \emph{fixed}
for a pole-node that is not allowed to form a nose. We maintain the
following properties through the recursive construction process: In
S- and R- nodes, $s$ and $t$ are not fixed. In P- and Q-nodes, only
one of them is fixed, say $s$. But similar to the 4-planar
biconnected case, we may swap their roles.

\begin{figure}[t]
    \centering
    \begin{minipage}[b]{.24\textwidth}
        \centering
        \subfloat[\label{fig:5p_layout}{}]
        {\includegraphics[width=.95\textwidth,page=1]{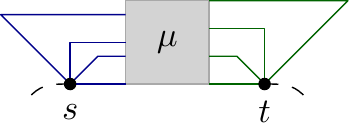}}
    \end{minipage}
    \begin{minipage}[b]{.24\textwidth}
        \centering
        \subfloat[\label{fig:5p_pnode}{}]
        {\includegraphics[width=.95\textwidth,page=3]{5p_bicon_new}}
    \end{minipage}
    \begin{minipage}[b]{.45\textwidth}
        \centering
        \subfloat[\label{fig:5p_snode}{}]
        {\includegraphics[width=.95\textwidth,page=2]{5p_bicon_new}}
    \end{minipage}
    \caption{
    (a)~Layout specification; $s$ and $t$ are located at the bottom.
    (b)~P-node with an $(s,t)$-edge from a Q-node $\mu_1$. $s$ and $t$ form a nose in $\mu_2, \mu_3$.
    (c)~S-node example with four children $\mu_1, \ldots,  \mu_4$.}
    \label{fig:5p_bicon}
\end{figure}

\begin{description}
\item[P-node case:] Let $\mu$ be a P-node. It is not difficult to see
that $\mu$ has at most $4$ children; one of them might be a Q-node,
i.e., an $(s,t)$-edge, which can be drawn at the bottom as a
horizontal segment. Since P-nodes are not adjacent to each other in
$\mathcal{T}$, the remaining children are S- or R-nodes. By our
invariant we may form noses enabling us to stack them as in
Fig.\ref{fig:5p_pnode}, as $s$ and $t$ are not fixed in them.
\item[S-node case:] Let $\mu$ be an S-node with children $\mu_1,
\ldots, \mu_l$. Instead of the diagonal layout used earlier, we now
align the drawings horizontally; see Fig.\ref{fig:5p_snode}. In the
S-node case, the poles inherit their pertinent degree from the
children and the same holds for the property of being fixed.
However, by our new invariant this is forbidden, as it clearly
states that $s$ and $t$ are not fixed. It is easy to see that when
$\mu_1$ is a P-node, $s$ is fixed by the invariant in $\mu_1$. In
this case, we swap the roles of the poles in $\mu_1$ such that $s$
is not fixed. However, the other pole of $\mu_1$, say $v_1$, is
fixed now. Since the skeleton of an S-node is a cycle of length at
least three, $v_1 \neq t$ holds. As a result, both $s$ and $t$ are
not fixed in the resulting drawing.
\item[R-node case:] To compute a layout of an R-node, we employ the
triconnected algorithm (with $s = v_1$ and $t =v_2$). So, let $\mu$
be an R-node and $\mu_{e}$ a child of $\mu$ that corresponds to
the virtual edge $e = (u,v)$ in $\skel{\mu}$. Then, $\pdeg{\mu_{e}}{u}
\leq 3$ and $\pdeg{\mu_{e}}{v} \leq 3$ holds. When inserting the
drawing of $\pert{\mu_e}$, we require at most three consecutive
ports at $u$ and $v$ for the additional edges. As the triconnected
algorithm assigns ports in a consecutive manner based on the
relative position of the endpoints, we modify the port assignment so
that an edge may have more than one port assigned. To do so, we
assign each edge $e = (u,v)$ in $\skel{\mu}$ a pair
$(\pdeg{\mu_e}{u}, \pdeg{\mu_e}{v}) \in \{ 1, 2, 3 \}^2$ that
reflects the number of ports required by this edge at its endpoints.
Then, we extend the triconnected algorithm such that when a port of
$u$ is assigned to an edge $e = (u,v)$, $\pdeg{\mu_e}{u}-1$
additional consecutive ports in clockwise or counterclockwise order
are reserved. The direction depends on the different types of edges
that we will discuss next.

\begin{figure}[t]
    \centering
     \begin{minipage}[b]{.16\textwidth}
        \centering
        \subfloat[\label{fig:5p_port_reserv_chain}{}]
        {\includegraphics[width=.9\textwidth,page=6]{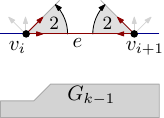}}
    \end{minipage}
    \begin{minipage}[b]{.16\textwidth}
        \centering
        \subfloat[\label{fig:5p_port_reserv_chain_inserted}{}]
        {\includegraphics[width=.9\textwidth,page=7]{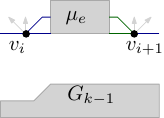}}
    \end{minipage}
    \begin{minipage}[b]{.16\textwidth}
        \centering
        \subfloat[\label{fig:5p_port_reserv_leg}{}]
        {\includegraphics[width=.9\textwidth,page=4]{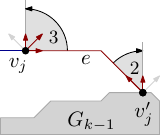}}
    \end{minipage}
    \begin{minipage}[b]{.16\textwidth}
        \centering
        \subfloat[\label{fig:5p_port_reserv_leg_inserted}{}]
        {\includegraphics[width=.9\textwidth,page=5]{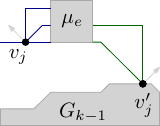}}
    \end{minipage}
    \begin{minipage}[b]{.16\textwidth}
        \centering
        \subfloat[\label{fig:5p_port_reserv_singleton}{}]
        {\includegraphics[width=.9\textwidth,page=8]{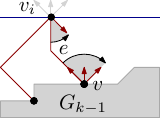}}
    \end{minipage}
    \begin{minipage}[b]{.16\textwidth}
        \centering
        \subfloat[\label{fig:5p_port_reserv_singleton_inserted}{}]
        {\includegraphics[width=.9\textwidth,page=9]{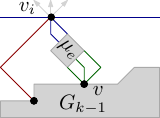}}
    \end{minipage}
    \caption{
    (a)~Virtual edge $e = (v_{i}, v_{i+1})$ connecting two consecutive vertices of a chain. At both endpoints the drawing of $\mu_e$ requires two ports.
    (b)~Replacing $e$ in (a) with the corresponding drawing of the child $\mu_e$.
    (c)~Example of an edge $e = (v_{j},v_{j}')$ that requires three ports at $v_{j}$ and two at  $v_{j}'$.
    (d)~Inserting the drawing of $\mu_e$ into (c) with $v_{j}$ being fixed and $v_{j}'$ forming a nose.
    (e)~Reserving ports for the nested edges. A single port for a real edge is reserved and then two ports for the virtual edge e = $(v_i, v)$.
    (f)~Final layout after inserting the drawing of $\mu_e$.}
    \label{fig:5p_bicon_R}
\end{figure}

The simplest type of edges are the ones among consecutive vertices
$v_i, v_{i+1}$ of a chain. Recall that $P_0 = \{ v_1, v_2 \} $ is a
special case and the edge $(v_1, v_2)$ is drawn differently. Also,
the edges from $P_0$ to $P_1$ are drawn as horizontal segments; see
Fig.\ref{fig:5p_before_final}. For each such edge we reserve the
additional ports at $v_i$ in counter-clockwise order and at
$v_{i+1}$ in clockwise order; see
Fig.\ref{fig:5p_port_reserv_chain}. So, we can later plug the
drawing of the children into the layout as in
Fig.\ref{fig:5p_port_reserv_chain_inserted} without forming noses.
The second type of edges are the ones that connect $P_k =
\left\{v_i,\ldots,v_j\right\}$ to $v_{i}'$  and $v_{j}'$ in
$G_{k-1}$. No matter if $P_k$ is a singleton or a chain, we proceed
by reserving the ports as in the previous case, i.e., at $v_i$
clockwise, ($v_j$ counter-clockwise, resp.) and at $v_{i}'$
counter-clockwise ($v_{j}'$ clockwise); see
Fig.\ref{fig:5p_port_reserv_leg}. In case where $(v_i, v_{i}')$ or
$(v_j, v_{j}')$ is a virtual edge, we choose the poles such that
$v_i$ ($v_j $ resp.) is fixed in $\mu_{(v_i, v_{i}')}$ ($\mu_{(v_j,
v_{j}')}$ resp.). Thus, we can create a nose with $v_{i}'$ ($v_{j}'$
resp.). Having exactly the ports required at both endpoints, we
insert the drawing by replacing the bend with a nose as in
Fig.\ref{fig:5p_port_reserv_leg_inserted}. The remaining edges from
$P_k$ to $G_{k-1}$ in case of a singleton $P_k = \{ v_i \}$ can be
handled similarly; see Fig.\ref{fig:5p_bicon_R}. Notice that during
the replacement of the edges, the fixed vertex is always the upper
one. The only exception are the horizontal drawn edges of a chain.
There, it does not matter which one is fixed, as none of the poles
has to form a nose.
\item[Root case:] We root $\mathcal{T}$ at an arbitrarily chosen
Q-node representing a real edge $(s,t)$. By our invariant we may
construct a drawing with $s$ and $t$ at the bottom of the drawing's
bounding box, hence, we draw the edge $(s,t)$ below the bounding box
with a ninety degree bend using the south east port at $s$ and south
west port at $t$.
\end{description}

\begin{theorem}
Given a biconnected 5-planar graph $G$, we can compute in $O(n^2)$
time an octilinear drawing of $G$ with at most one bend per edge.
\end{theorem}
\begin{proof}
We have shown that the ability to rotate and scale suffices to
extend the result from 4-planar to 5-planar at the expense of the
area. Similar to the 4-planar case, computing $\mathcal{T}$ takes
linear time. Hence, the overall runtime is governed by the
triconnected algorithm.
\end{proof}

%=================================================================
\subsection{The Simply Connected Case}
%=================================================================

In the following, we only outline the differences in comparison with
the corresponding 4-planar case. As an invariant, the drawing of
every subtree should conform to the layout depicted in
Fig.\ref{fig:5p_bc_tree_layout}. For a single biconnected component
$b$, let $v_c$ refer to the cut vertex linking it to the parent. As
root for the SPQR-tree $\mathcal{T}_b$ of $G_b$, we again choose a
Q-node $\mu_r$ whose real edge is incident to $v_c$; see
Fig.\ref{fig:5p_bc_t_root}. Hence, the layout generated by the
biconnected approach matches this scheme.

\begin{figure}[htb]
    \centering
     \begin{minipage}[b]{.24\textwidth}
        \centering
        \subfloat[\label{fig:5p_bc_tree_layout}{}]
        {\includegraphics[width=.9\textwidth,page=1]{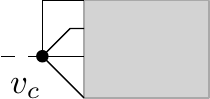}}
    \end{minipage}
     \begin{minipage}[b]{.24\textwidth}
        \centering
        \subfloat[\label{fig:5p_bc_t_root}{}]
        {\includegraphics[width=.9\textwidth,page=2]{5p_bc}}
    \end{minipage}
     \begin{minipage}[b]{.24\textwidth}
        \centering
        \subfloat[\label{fig:5p_bc_tree_s_ports}{}]
        {\includegraphics[width=.9\textwidth,page=3]{5p_bc}}
    \end{minipage}
      \begin{minipage}[b]{.24\textwidth}
        \centering
        \subfloat[\label{fig:5p_bc_tree_single_ports}{}]
        {\includegraphics[width=.9\textwidth,page=4]{5p_bc}}
    \end{minipage}
    \caption{
    (a)~Layout scheme for a BC-subtree rooted at $v_c$.
    (b)~Rooting $\mathcal{T}_b$ at a Q-node $\mu_r$.
    (c)~Attaching a subtree at a chain and in (d)~at a singleton inside an R-node.}
    \label{fig:5p_bc_tree}
\end{figure}

It remains to show that we can attach the children. Since we are
able to scale and rotate, we keep things simple and look for
suitable spots to attach them. Notice that in the drawings of
S-nodes and chains in R-nodes all southern ports are free. Hence, we
may rotate the drawings of the subtrees and attach the at most three
(two for a chain) edges to $v_c$ there (refer to
Fig.\ref{fig:5p_bc_tree_s_ports} for an example of a chain). The
only exception are the singletons. Assume that $v_i$ is an ordinary
singleton that has one nested edge attached. Hence, it has degree
four, leaving us with a single bridge to attach the component;
Fig.\ref{fig:5p_bc_tree_single_ports}. However, this does not hold
in case $v_i = v_n$. Consider the case where $v_n$ has a nested edge
and we have to attach a subtree that requires two ports. As a result
$v_n$ has degree $3$ in $G_b$ and, thus, all northern ports are
free.

\begin{theorem}
Given a connected 5-planar graph $G$, we can compute in $O(n^2)$
time an octilinear drawing of $G$ with at most one bend per edge.
\label{thm:5planarconnected}
\end{theorem}
\begin{proof}
We described how to attach any subtree to cut vertices inside a
biconnected component. Furthermore, the component itself complies
with the layout scheme. In addition, this scheme enables us to
compose such drawings at a cut vertex using rotations.
\end{proof}

%=================================================================
\section{A Note on Octilinear Drawings of 6-Planar Graphs}
\label{sec:6planar}
%=================================================================
In this section, we show that it is not always possible to construct
a planar octilinear drawing of a given $6$-planar graph with at most
one bend per edge. In particular, we present an infinite class of
6-planar graphs, which do not admit planar octilinear drawings with
at most one bend per edge.

\begin{theorem}
There exists an infinite class of 6-planar graphs which do not admit
planar octilinear drawings with at most one bend per edge.
\label{thm:6planar}
\end{theorem}
\begin{proof}
Our proof is heavily based on the following simple observation: If
the outer face $\mathcal{F}(\Gamma(G))$ of a given planar octilinear
drawing $\Gamma(G)$ consists of exactly three vertices, say $v,v'$
and $v''$, that have the so-called \emph{outerdegree-property},
i.e., $deg(v) = deg(v') = 6$ and $5 \leq deg(v'') \leq 6$, then it
is not feasible to draw all edges delimiting
$\mathcal{F}(\Gamma(G))$ with at most one bend per edge; one of them 
has to be drawn with (at least) two bends in $\Gamma(G)$. Next, we
construct a specific maximal 6-planar graph, in which each face has
at most one vertex of degree $5$ and at least two vertices of degree
$6$; see Fig.\ref{fig:6p_twobends}. This specific graph does not
admit a planar octilinear drawing with at most one bend, as its
outerface is always bounded by three vertices that have the
outerdegree-property.

To obtain an infinite class of 6-planar graphs with this property,
we give the following recursive construction. Let $G_1$ and $G_2$ be
two copies of the graph of Fig.\ref{fig:6p_twobends}. Let also $f_i$
be a bounded face of $G_i$, $i=1,2$. We proceed to subdivide each
edge of $f_i$ by introducing a new vertex on it. We further assume
that the new vertices of $f_i$ are pairwise adjacent (see the top
part of Fig.\ref{fig:6p_construction}). Hence, they form a
triangular face, say $f_i'$, in the augmented graph, say
$G_i^{aug}$, constructed in this manner. Up to now, each of the
newly introduced vertices is of degree four. Now, assume that
$G_1^{aug}$ is drawn on the plane so that $f_1'$ is a bounded face
in $\Gamma(G_1^{aug})$, and $G_2^{aug}$ is drawn such that $f_2'$ is
the unbounded face in $\Gamma(G_2^{aug})$. By choosing $f_2'$ as the
outer face in $\Gamma(G_2^{aug})$ we make sure that we can connect
the three degree four vertices of $f_2'$ to the three degree four
vertices of $f_1'$ in the following way: We appropriately scale down
$\Gamma(G_2^{aug})$ and proceed to draw it in the interior of $f_1'$
without introducing any crossings (see the small gray-colored
triangle of the bottom drawing of Fig.\ref{fig:6p_construction}). If
we connect the vertices of $f_1'$ and $f_2'$ in an ``octahedron-like
manner'', then all vertices of $f_1'$ and $f_2'$ are of degree $6$
and the resulting graph, say $G_1^{aug} \oplus G_2^{aug}$, is
maximal $6$-planar and has the outerdegree-property.
\end{proof}

\begin{figure}[t]
    \centering
    \begin{minipage}[b]{.48\textwidth}
        \centering
        \subfloat[\label{fig:6p_twobends}{}]
        {\includegraphics[width=\textwidth,page=1]{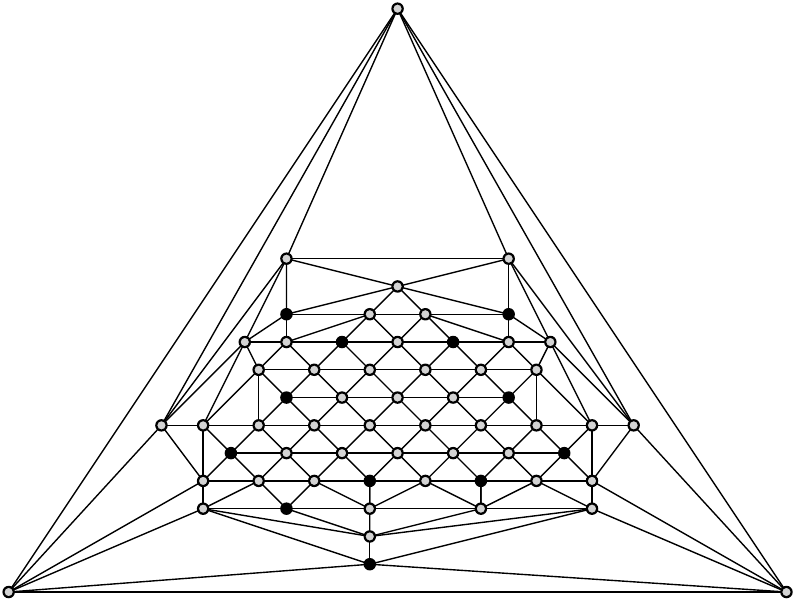}}
    \end{minipage}
    \begin{minipage}[b]{.48\textwidth}
        \centering
        \subfloat[\label{fig:6p_construction}{}]
        {\includegraphics[width=\textwidth,page=2]{6p_twobends}}
    \end{minipage}
    \caption{(a)~A maximal $6$-planar graph in which each face has at most one vertex of degree $5$ (black-colored vertices) and at least two vertices of degree $6$ (gray-colored vertices).
    From Euler's formula for maximal planar graphs, it follows that any graph with this property must have at least $12$ vertices of degree $5$.
    Hence, this is the smallest graph with this property.
    (b)~Illustration of the recursive construction.}
    \label{fig:6p_twobends_construction}
\end{figure}

%=================================================================
\section{A Sample Octilinear Drawing with at most 1 bend per edge}
\label{sec:sample}
%=================================================================

\begin{figure}[h!]
    \centering
    \includegraphics[width=\textwidth]{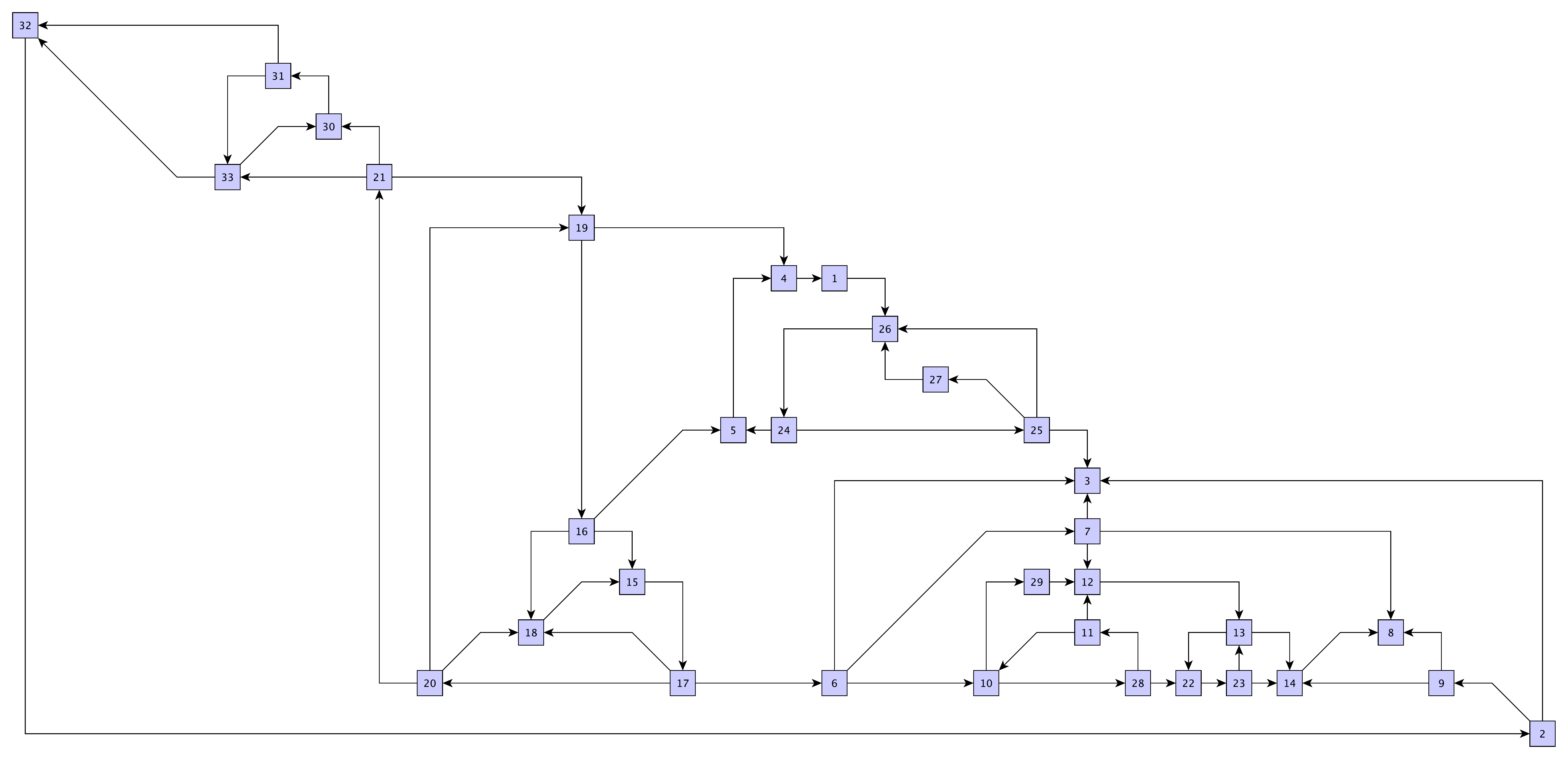}
    \caption{Example layout of a biconnected 4-planar graph. Vertices are labeled by their indices. The
    corresponding SPQR-tree $\mathcal{T}$ has been rooted at a Q-node
    representing the edge $(v_{32}, v_{2})$ with the only child being an
    S-node whose skeleton is the simple cycle $v_{32}, v_{21}, v_{2}$.
    It has two R-nodes as children, a smaller in the upper left (with
    poles $\{ v_{32}, v_{21} \}$) and a larger one (with poles $ \{
    v_{21}, v_{2} \}$) occupying most of the drawing area. The latter
    one contains two smaller S-nodes (with poles $\{ v_{10}, v_{12} \}$
    and $\{ v_{4}, v_{26} \}$) and a P-node (with poles $ \{
    v_{26},v_{25} \}$) that has two children. One of them being an
    $(s,t)$-edge, the other one an S-node.}
    \label{fig:4p_example_large}
\end{figure}

%=================================================================
\section{Conclusions}
\label{sec:conclusions}
%=================================================================

Motivated by the fundamental role of planar octilinear
drawings in map schematization, we presented algorithms for their construction
with at most one bend per edge for 4- and 5-planar graphs.
We also improved the known bounds on the required number of slopes
for $4$- and $5$-planar drawings from $8$ and $10$, resp.
(\cite{KPP13}) to $4$. Our work raises several open problems:

\begin{itemize}
\item Is it possible to construct planar octilinear drawings of
4-planar (5-planar, resp.) graphs with at most one bend per edge in
$o(n^3)$ (polynomial, resp.) area?
\item Does any triangle-free 6-planar graph admit a planar
octilinear drawing with at most one bend per edge?
\item What is the complexity to determine whether a $6$-planar graph
admits a planar octilinear drawing with at most one bend per edge?
\item What is the number of necessary slopes for bendless drawings
of $4$-planar graphs?
\end{itemize}

%=================================================================
\bibliographystyle{abbrv}
\bibliography{references}
%=================================================================

\end{document}